\documentclass[11pt]{article}
\usepackage{amsmath, amssymb, amsthm, epsfig, verbatim, xcolor}

\newtheorem{theorem}{Theorem}[section]
\newtheorem{lemma}[theorem]{Lemma}
\newtheorem{proposition}[theorem]{Proposition}

\newtheorem{corollary}[theorem]{Corollary}
\newtheorem{conjecture}[theorem]{Conjecture}
\newtheorem{definition}{Definition}[section]

\newcommand{\supp}{\operatorname{supp}}

\def\ds{\displaystyle}
\def\lt{\left}
\def\rt{\right}

\begin{document}

\title{\textbf{On the maximum number of minimal codewords}}
\author{Romar dela Cruz\textsuperscript{1} and Sascha Kurz\textsuperscript{2}\\\\
\small{\textsuperscript{1}Institute of Mathematics, University of the Philippines Diliman, Philippines}\\
\small{\textsuperscript{2}Department of Mathematics, University of Bayreuth, Germany}}
\date{}

\maketitle

\begin{abstract}
  Minimal codewords have applications in decoding linear codes and in cryptography. We study the maximum number of minimal codewords in binary linear codes 
  of a given length and dimension. Improved lower and upper bounds on the maximum number are presented. We determine the exact values for the case of linear 
  codes of dimension $k$ and length $k+2$ and for small values of the length and dimension. We also give a formula for the number of minimal codewords of 
  linear codes of dimension $k$ and length $k+3$. 
\end{abstract}


\section{Introduction}
The minimal codewords of a linear code are those whose supports, i.e., the set of nonzero coordinates, do not properly contain the supports of other nonzero codewords.  
They are equivalent to circuits in matroids and cycles in graphs. In coding theory, minimal codewords were first used in decoding algorithms \cite{Agrell,Agrell2,AB,Hwang}. 
They have also found applications in cryptography: in secret sharing schemes \cite{Massey} and in secure two-party computation \cite{CCP}.

The set of minimal codewords is only known for a few classes of codes (see \cite{Agrell, ABN, AB, BM, DKL, DY, kurz2020number, SST, TQLZ,YF}) and, in general, it is a very hard problem to 
determine this set. In this work, we consider the following question: what is the maximum number of minimal codewords of linear codes of a given length and dimension? 
This problem is already studied in the case of cycles in graphs \cite{ES}. In the matroid setting, the maximum number of circuits was first addressed in \cite{DSL}. 
The study of the maximum and minimum number of minimal codewords of linear codes was initiated in \cite{Minmin,Maxmin2,Maxmin,CKKW}. 

The results in this paper are described as follows. We determine the maximum number of minimal codewords for binary linear codes of dimension $k$ and length $k+2$. We 
also give a formula for the number of minimal codewords for the case of dimension $k$ and length $k+3$. A general construction of linear codes with a relatively large 
number of minimal codewords is also presented. This gives a lower bound that is asymptotically close to the matroid upper bound. An upper bound that is better than the 
matroid upper bound is also derived. The key idea is to use the systematic generator matrix for a linear code and analyze the properties of the subsets of rows that 
produce minimal codewords. We also compute the exact values of the for maximum number of minimal codewords small values of length and dimension (completing the table 
in \cite{Maxmin2}).

\section{Preliminaries}

Let $q$ be a power of a prime $p$ and $\mathbb{F}_q$ be the finite field of order $q$. 
An $[n,k]_q$ linear code $C$ is a $k$-dimensional subspace of $\mathbb{F}_q^n$. Given a vector $x\in\mathbb{F}_q^n$, the support of $x$ is defined as 
$\supp(x)=\{i\,:\, x_i\neq 0, 1\leq i\leq n \}$. 
A $k\times n$ matrix $G$ whose rows form a basis for $C$ is called a generator matrix. If $G=[I_k|A]$, where $I_k$ is the $k\times k$ identity matrix, then we say 
that $G$ is systematic or in standard form.

A nonzero codeword $c\in C$ is minimal if there does not exist a nonzero codeword $c'$ such that $\supp(c')\subset\neq \supp(c)$. 
Otherwise (including the case $c=\mathbf{0}$), we call the codeword $c$ non-minimal. General properties of minimal codewords 
can be found in \cite{AB}. Note that a codeword and its nonzero scalar multiples have the same support.  We say that two codewords are equivalent if one is a scalar multiple 
of the other.  We use the notation $M(C)$ for the number of non-equivalent minimal codewords of $C$. Let $M_q(n,k)$ be the maximum of $M(C)$ for all $[n,k]_q$ codes $C$. 
Since $C$ has $q^k-1$ nonzero codewords, we have$$M_q(n,k)\leq \dfrac{q^k-1}{q-1}.$$

Bounds for $M_q(n,k)$ and some exact values can be found in \cite{Agrell2, Maxmin2, AB, Maxmin,CKKW,DSL}. In the setting of matroids, it was shown in \cite{DSL}, that 
\begin{equation}
  M_q(n,k)\leq {n\choose k-1}. \label{mub}
\end{equation}
This is bound is also called the matroid upper bound. Alternative proofs were given in \cite{Maxmin}. Inequality~(\ref{mub}) is satisfied with equality for MDS codes. 
Another upper bound was derived by Agrell in \cite{Agrell2} for binary codes with high rate: for $\frac{k-1}{n}>\frac{1}{2}$, we have 
$$M_2(n,k)\leq \frac{2^k}{4n\lt(\dfrac{k-1}{n}-\dfrac{1}{2}\rt)}.$$

Based on random coding, the lower bound
$$
  M_q(n,k)\geq \ds\sum_{j=0}^{n-k+1} {n\choose j}\dfrac{(q-1)^j}{q^{n-k}}\prod_{i=0}^{j-2} \lt[1-q^{-(n-k-i)}\rt]
$$ 
was given in \cite{AB}.

It is clear that we have $M_q(n,1)=1$ and $M_q(k,k)=k$ for all $k\geq 1$. In \cite{Maxmin2}, it was shown that $M_2(k+1,k)={k+1\choose 2}$ for $k\geq 2$. For small 
values of $k$ and $n$, the authors in \cite{Maxmin2} presented some exact values and bounds on $M_2(n,k)$. In addition, exact values for the case of cycle codes were obtained. 

\section{Relations between minimal codewords and the rows of a systematic generator matrix}
Let $C$ be a linear $[k+t,k]_2$, i.e.\ binary, code with systematic generator matrix $G$. By $g^i$ we denote the $i$th row of $G$, where $1\le i\le k$. For each subset 
$S\subseteq\{1,\dots,k\}$ let $c^S$ denote the sum of the rows of $G$ with indices in $S$, i.e., $c^S=\sum_{i\in S}g^i\in C$. For each codeword $c\in C$ let $c_S\in\mathbb{F}_2^k$ 
denote the systematic part of $c$, i.e., the restriction of $c$ to the first $k$ coordinates $c_1,\dots, c_k$. Similarly, for each codeword $c\in C$ let $c_I\in\mathbb{F}_2^t$ 
denote the information bits, i.e., the restriction of $c$ to the last $t$ coordinates $c_{k+1},\dots,c_{k+t}$. Some of the subsequent observations can also be found in \cite{kurz2020number}. 

\begin{lemma}
  \label{lemma_zero_sum}
  Let $\emptyset\neq S\subseteq\{1,\dots,k\}$. If there exists a subset $\emptyset\neq T\subsetneq S$ with $c^T_I=\mathbf{0}$, then $c^S$ is non-minimal.
\end{lemma} 
\begin{proof}
  Since $\supp\!\left(c^{S\backslash T}_I\right)=\supp\!\left(c^{S}_I\right)$ and $\supp\!\left(c^{S\backslash T}_S\right)\subsetneq \supp\!\left(c^{S}_S\right)$, we have 
  $\supp\!\left(c^{S\backslash T}\right)\subsetneq \supp\!\left(c^{S}\right)$.
\end{proof}

\begin{lemma}
  Let $\emptyset\neq S\subseteq\{1,\dots,k\}$. The codeword $c^S$ is non-minimal iff there exists a subset $\emptyset\neq T\subsetneq S$ with $\supp(c^T_I)\subseteq \supp(c^S_I)$.
\end{lemma}
\begin{proof}
  Since $S\neq \emptyset$ we have $c^S\neq \mathbf{0}$. Thus, if $c^S$ is non-minimal, there exists a subset $\emptyset\neq T\subsetneq S$ with $\supp(c^T)\subsetneq \supp(c^S)$, 
  so that $\supp(c^T_I)\subseteq \supp(c^S_I)$. For the other direction let $\emptyset\neq T\subsetneq S$ with $\supp(c^T_I)\subseteq \supp(c^S_I)$. If $\supp(c^T_I)\neq \supp(c^S_I)$, 
  then $\supp(c^T_I)\subsetneq \supp(c^S)$ implies $\supp(c^T)\subsetneq \supp(c^S_I)$ so that $c^S$ is non-minimal by definition. If $\supp(c^T_I)= \supp(c^S_I)$, then 
  $c^{S\backslash T}_I=\mathbf{0}$ and we can apply Lemma~\ref{lemma_zero_sum}.
\end{proof}  

\begin{corollary}
  \label{cor_max_card_S}
  Let $c^S$ be a minimal codeword. Then, we have $1\le \#S\le t+1$. Moreover, if $\# S=t+1$, then $c^S_I=\mathbf{0}$.   
\end{corollary}
\begin{proof}
  The largest cardinality of a set of linearly independent vectors in $\mathbb{F}_2^t$ is $t$. Thus, if $\#S\ge t+1$, then there exists a subset $T\subseteq S$ with 
  $c^T_I=\mathbf{0}$ and $\#T\le t+1$. We finally apply Lemma~\ref{lemma_zero_sum} to conclude $\#S\le t+1$.  
\end{proof}

As a direct consequence we conclude
$$
  M_2(k+t,k)\le \sum_{i=1}^{t+1} {{k+t}\choose i},
$$
which asymptotically tends to ${{k+t}\choose{t+1}}$ for a fixed value of $t$ (if $k$ tends to infinity). In Proposition~\ref{prop_first_improved_upper_bound} 
we will present a strict improvement over the matroid upper bound ${n\choose{k-1}}={{k+t}\choose{t+1}}$, see (\ref{mub}), provided that $k$ is large enough. 

\begin{lemma}
  \label{lemma_zero_sum_characterization}
  Let $\emptyset\neq S\subseteq\{1,\dots,k\}$ be a subset such that $c^S_I=\mathbf{0}$. Then, $c^S$ is minimal iff $c^T_I\neq\mathbf{0}$ for all $\emptyset\neq T\subsetneq S$.
\end{lemma}
\begin{proof}
  Since $S\neq \emptyset$ we have $c^S\neq \mathbf{0}$. If $c^S$ is non-minimal, then there exists a subset $\emptyset\neq T\subsetneq S$ with $\supp(c^T)\subsetneq\supp(c^S)$. 
  Since $c^S_I=\mathbf{0}$ this implies $c^T_I=\mathbf{0}$. For the other direction we apply Lemma~\ref{lemma_zero_sum}. 
\end{proof}

\begin{lemma}
  Let $G$ be a systematic generator matrix of an $[k+t,k]_2$ code $C$ and $1\le i\le k$ be an index with $g^i_I=\mathbf{0}$. By $G'$ we denote the matrix that arises from 
  $G$ by removing the $i$th row $g^i$ and by $G''$ the matrix if we additionally remove the $i$th column. Let $C'$ and $C''$ be the linear codes generated by $G'$ and 
  $G''$, respectively. Then $C'$ is $[k+t,k-1]_2$ code, $C''$ a $[k+t-1,k-1]_2$ codes, and we have $M(C)=M(C')+1=M(C'')+1$.  
\end{lemma}
\begin{proof}
  The stated lengths and the dimensions of the codes $C'$ and $C''$ directly follow from their construction. Since removing a zero column in a generator matrix does not 
  change the number of minimal codewords, we have $M(C')=M(C'')$, so that it remains to show $M(C)=M(C')+1$. The codeword $g^i$ itself is minimal in $C$ and not contained 
  in $C'$. For any subset $\{i\}\subsetneq S\subseteq\{1,\dots,k\}$ the codeword $c^S$ is non-minimal due to Lemma~\ref{lemma_zero_sum} (choosing $T=\{i\}$). 
  It remains to show that for subsets $\emptyset\neq S\subseteq\{1,\dots,k\}\backslash\{i\}$ the codeword $c^S\in C'\le C$ is minimal in $C'$ iff it is minimal in $C$. Since 
  $C'$ is a subcode of $C$ we only need to consider the case where $c^S$ is non-minimal in $C$. Then, there exists a subset $\emptyset\neq T\subsetneq S$ with $\supp(c^T)\subsetneq 
  \supp(c^S)$. Since $g^i_I=\mathbf{0}$ we can assume $i\notin T$, so that $c^T\in C'$ and $c^S$ is also non-minimal in $C'$. 
\end{proof}

So, in the following we may assume $c^S_I\neq\mathbf{0}$ whenever needed and we mention the implication $M_2(k,k)=k$ for all $k\ge 1$. 

\begin{definition}
  Let $C$ and especially $t$ be given. By $\mathcal{T}$ we denote the set of the $2^t$ elements of $\mathbb{F}_2^t$. For each 
  $\tau\in\mathcal{T}$ we set $a_\tau=\#\left\{1\le i\le k\,:\, r^i_I=\tau\right\}$. The counting vector of all $a_\tau$ is denoted by 
  $\mathbf{a}$. More precisely, we write $a_\tau(C)$ and $\mathbf{a}(C)$ whenever the code $C$ is not clear from the context.
\end{definition}  

Since column and row permutations of a generator matrix do not change the number of minimal codewords, we have:
\begin{lemma}
  Let $C$ and $C'$ be two $[k+t,k]_2$ codes. If $\mathbf{a}(C)=\mathbf{a}(C')$, then $M(C)=M(C')$.
\end{lemma}

For the case $t=1$ we can easily determine $M(C)$ given the vector $\mathbf{a}(C)=(a_0,a_1)$.
\begin{lemma}
  Let $C$ be a $[k+1,k]_2$ code. Then, $M(C)=k+{a_1\choose 2}$.
\end{lemma}
\begin{proof}
  For all subsets $S\subseteq\{1,\dots,k\}$ of cardinality $1$, the codeword $c^S$ is minimal, which give $k$ minimal codewords. Due to Corollary~\ref{cor_max_card_S} 
  is suffices to consider codewords of the form $c^S$ with $\emptyset\subseteq S\subseteq \{1,\dots,k\}$ and $\# S\le 2$, so that it remains to consider the cases with 
  $\#S=2$. Due to Lemma~\ref{lemma_zero_sum}, Corollary~\ref{cor_max_card_S}, and Lemma~\ref{lemma_zero_sum_characterization} the codeword $c^{\{i,j\}}$ is minimal iff 
  $i\neq j$ and $g^i_I=g^j_I=1$. 
\end{proof}

\begin{corollary}
  $M_2(k+1,k)={{k+1}\choose 2}=(k+1)k/2$.
\end{corollary}
The same result was also obtained in \cite{Maxmin2}. Not that the matroid upper bound $M_2(n,k)\leq {n\choose k-1}={{k+t}\choose{k-1}}={{k+t}\choose{t+1}}$, see (\ref{mub}), 
is matched with equality.  We remark that the unique code attaining this upper bound is the so-called projective base (of $\mathbb{F}_2^k$) given by a generator matrix 
consisting of the $k$ unit vectors and the all-$1$-vector as columns.

In Lemma~\ref{lemma_zero_sum_characterization} we have characterized whether $c^S$ is minimal for the special case when $c^S_I=\mathbf{0}$ using the information bits of $g^i$, 
where $i\in S$, only. This can be generalized and formalized as follows.

\begin{definition}
  \label{def_reduced_code}
  Let $C$ be a $[k+t,k]_2$ code and $\emptyset\neq S\subseteq\{1,\dots,k\}$ a subset. With this, we set
  $$
    \overline{C^S}:=\left\langle\left\{g^i_I\,:\,i\in S\right\}\right\rangle.
  $$
  We call $\overline{C^S}$ the reduced code of $C$ with respect to $S$.
\end{definition}  
 
\begin{lemma}
  \label{lemma_reduced_code}
  Let $C$ be a $[k+t,k]_2$ code and $\emptyset\neq S\subseteq\{1,\dots,k\}$ a subset. The codeword $c^S$ is minimal in $C$ iff $c^S_I$ is either minimal in 
  $\overline{C^S}$ or $c^S_I=\mathbf{0}$ and $c^T_I\neq \mathbf{0}$ for all $\emptyset\neq T\subsetneq S$.
\end{lemma}
\begin{proof}
  Assume that $c^S$ is non-minimal. Since $S\neq \emptyset$ we have $c^S\neq\mathbf{0}$, so that there exists a subset $\emptyset\neq T\subsetneq S$ with 
  $\supp\!\left(c^T\right)\subsetneq \supp\!\left(c^S\right)$. Thus, we have $\supp\!\left(c^T_I\right)\subseteq \supp\!\left(c^S_I\right)$. If 
  $\supp\!\left(c^T_I\right)\neq  \supp\!\left(c^S_I\right)$, then $\supp\!\left(c^T_I\right)\subsetneq \supp\!\left(c^S_I\right)$ and $c^S_I$ is non-minimal in 
  $\overline{C^S}$. If $\supp\!\left(c^T_I\right)= \supp\!\left(c^S_I\right)$, then $c^{S\backslash T}_I=\mathbf{0}$. So, either $c^S_I\neq\mathbf{0}$ and 
  $\mathbf{0}=c^{S\backslash T}_I\subsetneq c^S_I$ or $c^S_I=\mathbf{0}$ and $c^{S\backslash T}_I=\mathbf{0}$, where $\emptyset\neq S\backslash T\subsetneq S$.
  
  For the other direction we first assume that $c^S_I$ is non-minimal in $\overline{C^S}$ and $c^S_I\neq \mathbf{0}$. Here, there exists a subset $\emptyset\neq T
  \subsetneq S$ with $\supp\!\left(c^T_I\right)\subsetneq \supp\!\left(c^S_I\right)$, which implies $\supp\!\left(c^T\right)\subsetneq \supp\!\left(c^S\right)$, i.e., 
  $c^S$ is non-minimal in $C$. In the other case we assume $c^S_I=\mathbf{0}$ and the existence of a subset $\emptyset \neq T\subsetneq S$ with $c^T_I=\mathbf{0}$. 
  Here we have $\supp\!\left(c^T\right)\subsetneq \supp\!\left(c^S\right)$, i.e., $c^S$ is non-minimal.       
\end{proof}

\begin{definition}
  \label{def_minimal_generating}
  We call a subset $\hat{S}\subseteq\mathbb{F}_2^t$ minimal generating if $\sum_{x\in\hat{S}} x$ is minimal in $\langle\hat{S}\rangle$ or 
  $\sum_{x\in\hat{S}} x=\mathbf{0}$ and $\sum_{x\in\hat{T}}x \neq\mathbf{0}$ for all $\emptyset\neq \hat{T}\subsetneq \hat{S}$.
\end{definition}

Note that no minimal generating set of cardinality at least two can contain the zero vector.

\begin{theorem}
  \label{thm_M_C_a_formula}
  Let $C$ be a linear $[k+t,k]_2$ code and $\mathbf{a}$ its corresponding vector counting the multiplicities of the occurring information vectors. With this, we have
  $$
    M(C)=k+\sum_{\tau\in\mathbb{F}_2^t\backslash\{\mathbf{0}\}} {{a_\tau}\choose 2}+\sum_{\hat{S}\subseteq\mathbb{F}_2^t\,:\,\hat{S}\text{ is minimal generating and } 2\le \#\hat{S}\le t+1} 
    \,\,\prod_{\tau\in\hat{S}} a_\tau.
  $$
\end{theorem}
\begin{proof}
  Let $c^S$ be a minimal codeword in $C$ for a subset $S\subseteq\{1,\dots,k\}$. Since $c^S\neq\mathbf{0}$ we have $S\neq \emptyset$. If $\# S=1$, then $c^S$ is minimal in 
  all cases, which gives $k$ possibilities. If $S$ contains two different elements $i$ and $j$ with $g^i_I=g^j_I$, then we deduce $\#S=2$ from Lemma~\ref{lemma_zero_sum} 
  and Lemma~\ref{lemma_zero_sum_characterization}. Since $i\neq j$ the codeword $c^{\{i,j\}}$ is indeed minimal, iff $g^i_I=g^j_I\neq\mathbf{0}$, which yields 
  $\sum_{\tau\in\mathbb{F}_2^t\backslash\{\mathbf{0}\}} {{a_\tau}\choose 2}$ further possibilities. In the remaining cases we have $2\le \#S\le t+1$, see Corollary~\ref{cor_max_card_S} 
  for the upper bound, and $g^i_I\neq g^j_I$ for all different $i,j\in S$. In other words $\hat{S}:=\left\{g^i_I\,:\, i\in S\right\}$ has cardinality $\# S$. 
  Due to Lemma~\ref{lemma_reduced_code} and Definition~\ref{def_minimal_generating} $c^S$ is minimal iff $\hat{S}$ is minimal generating. Given $\hat{S}$, the number of choices for $S$ are 
  $\prod_{\tau\in\hat{S}} a_\tau$.  
\end{proof}

In some cases it is possible to concretely describe the minimal generating sets in the formula of Theorem~\ref{thm_M_C_a_formula}:
\begin{proposition}
  \label{prop_canonical_projective_base}
  Let $C$ be a linear $[k+t,k]_2$ code and $\mathbf{a}$ its corresponding vector counting the multiplicities of the occurring information vectors. 
  If $a_\tau>0$ implies $\tau\in\mathcal{T}:=\left\{e_1,\dots,e_t,\mathbf{1}\right\}$, where $\mathbf{1}=e_1+\dots+e_t$, then we have
  $$  
    M(C)=k+\sum_{\tau\in\mathcal{T}} {{a_\tau}\choose 2}+\sum_{\mathbf{1}\subsetneq \hat{S}\subseteq\mathcal{T}} 
    \prod_{\tau\in\hat{S}} a_\tau. 
  $$
\end{proposition}
\begin{proof}
  Due to Theorem~\ref{thm_M_C_a_formula} it suffices to check which subsets of $\mathcal{T}$ are minimal generating. If $\mathbf{1}\notin\hat{S}$, then 
  $\sum_{x\in\hat{S}} x$ is clearly not minimal within $\hat{S}$. In all other cases $\hat{S}$ is minimal generating, which easily follows from Lemma~\ref{lemma_reduced_code}.  
\end{proof}     

As an example let $k\ge 2t$ be integers and $A$ be the $k\times t$ matrix whose rows consist of $2$ copies each of the unit vectors $e_1,\ldots,e_t$ and $k-2t$ copies of the zero vector.   
Consider the $[k+t,k]$ linear code $C$ with generator matrix $G=[I_k\,|\,A]$. Note that $C$ is projective and   
$$M(C)=k+\ds\sum_{\tau\in\{e_1,\ldots,e_t\}} {a_{\tau}\choose 2}=k+t.$$
In \cite[Lemma 5.1]{Kashyap} it is shown that each projective $[k+t,k]_2$ code $C$ satisfies $M(C)\ge k+t$.

\section{Bounds for the maximum number of minimal codewords}

A projective base can also be used to construct linear $[k+t,k]_2$ codes with a relatively large number of minimal codewords. To this end, let $e_i$ denote the $i$th 
unit vector and $\mathbf{1}$ denote the all-$1$-vector (in $\mathbb{F}_2^t$). 

\begin{proposition}
  \label{prop_projective_base_construction}
  $$
    M_2(k+t,k)\ge \left\lfloor\frac{k}{t+1}\right\rfloor^{t+1}
  $$
\end{proposition} 
\begin{proof}
  W.l.o.g.\ we assume $k\ge t+1$. Let $C$ be a linear $[k+t,k]_2$ code with systematic generator matrix $G$ such that $a_\tau(C)=0$ if 
  $\tau\notin\left\{e_1,\dots,e_t,\mathbf{1}\right\}$ and $a_\tau\ge \left\lfloor\frac{k}{t+1}\right\rfloor$ if $\tau\in\left\{e_1,\dots,e_t,\mathbf{1}\right\}$  
  for all $\tau\in\mathbb{F}_2^t$. Since $(t+1)\cdot \left\lfloor\frac{k}{t+1}\right\rfloor \le k$, 
  the construction is possible. Now we consider all subsets $S\subseteq \{1,\dots,k\}$ with cardinality $\#S=t+1$ such that $\#\left\{c^i_I\,:\,i\in S\right\}=t+1$, 
  i.e., each possible vector of information bits occurs exactly once. Note that there are 
  $$
    a_{e_1}\cdot\dots a_{e_t}\cdot a_{\mathbf{1}} \ge \left\lfloor\frac{k}{t+1}\right\rfloor^{t+1}
  $$  
  choices. Since $\sum_{i=1}^t e_i=\mathbf{1}$ and no proper subset of $\left\{e_1,\dots,e_t,\mathbf{1}\right\}$ sums to zero we can apply Lemma~\ref{lemma_zero_sum_characterization} 
  to deduce that those $c^S$ are minimal codewords. 
\end{proof}
The essential property of $\left\{e_1,\dots,e_t,\mathbf{1}\right\}$ used in the above proof is that of a projective basis. The explicit choice of vectors is called canonical 
basis in that context. We remark that it is also possible to precisely determine $M(C)$ if $a_\tau(C)\neq 0$ implies $\tau\in\left\{e_1,\dots,e_t,\mathbf{1}\right\}$ and those 
$a_\tau$ are given, see Proposition~\ref{prop_canonical_projective_base}. The codes constructed in Proposition~\ref{prop_projective_base_construction} show that the matroid upper 
bound $M_2(n,k)\leq {n\choose k-1}={{k+t}\choose{t+1}}$ is, up to a constant, asymptotically tight for every fixed value of $t$.

Our next aim is to conclude an upper bound for $M_2(k+t,k)$ from Theorem~\ref{thm_M_C_a_formula}. To this end, we will utilize an optimization problem\footnote{We 
are pretty sure that this problem has been studied in the literature before. However, since we were not able to find a reference, we give a self-contained proof here.}:
\begin{lemma}
  \label{lemma_symmetric_function_optimization}
  Let $s$, $r$, and $m$ be positive integers with $s\le r$ and $f\colon\mathbb{R}_{\ge 0}^r\to\mathbb{R}_{\ge 0}$ a function defined by
  $$
    f(x_1,\dots,x_r)=\sum_{S\subseteq \{1,\dots,r\}\,:\,\# S=s} \,\,\prod_{i\in S} x_i.
  $$
  Then, the optimization problem $\max f(x_1,\dots,x_r)$ subject to the constraint $\sum_{i=1}^r x_i=m$ has the unique optimal solution $x_i=\tfrac{m}{r}$ for all 
  $1\le i\le r$ with target value ${r\choose s}\cdot \left(\tfrac{m}{r}\right)^s$. If we additionally require that the $x_i$ have to be integers, then an 
  optimal solution is given by $x_i=\left\lfloor\tfrac{m+i-1}{r}\right\rfloor$ for $1\le i\le r$.
\end{lemma}
\begin{proof}
  For $r=1$ the statements are obvious, so that we assume $r\ge 2$ in the following. 
  Assume that for a given optimal solution of the real-valued optimization problem stated above, there are indices $1\le i,j\le r$ with $x_i\neq x_j$. From the given 
  vector $\mathbf{x}=(x_1,\dots,x_r)$ we construct a vector $\bar{\mathbf{x}}$ by replacing the $i$th and the $j$th component of $\mathbf{x}$ both by $\tfrac{x_i+x_j}{2}$. 
  Now we want to compare $f(\mathbf{x})$ and $f(\bar{\mathbf{x}})$. Clearly, we have
  $$
    \sum_{S\subseteq\{1,\dots,r\}\backslash\{i,j\}\,:\,\#S=s} \,\,\prod_{h\in S} \bar{x}_h= 
    \sum_{S\subseteq\{1,\dots,r\}\backslash\{i,j\}\,:\,\#S=s} \,\,\prod_{h\in S} x_h.
  $$  
  For the cases where the subset $S$ intersects $\{i,j\}$ in exactly one element we compute
  \begin{eqnarray*}
    && \sum_{S\subseteq\{1,\dots,r\} \,:\,\#S=s, \# S\cap\{i,j\}=1} \,\,\prod_{h\in S} \bar{x}_h \\ 
    &=& \sum_{\bar{S}\subseteq\{1,\dots,r\}\backslash\{i,j\}\,:\,\#\bar{S}=s-1} \,\,\left(\bar{x}_i+\bar{x}_j\right)\cdot\prod_{h\in \bar S} \bar{x}_h \\
    &=& \sum_{\bar{S}\subseteq\{1,\dots,r\}\backslash\{i,j\}\,:\,\#\bar{S}=s-1} \,\,\left(x_i+x_j\right)\cdot\prod_{h\in \bar S} x_h \\
    &=& \sum_{S\subseteq\{1,\dots,r\} \,:\,\#S=s, \# S\cap\{i,j\}=1} \,\,\prod_{h\in S} x_h,
  \end{eqnarray*}
  i.e., again there is no difference. If $S$ contains both $i$ and $j$, then we can write $S=\bar{S}\cup\{i,j\}$ with a subset $\bar{S}\subseteq \{1,\dots,r\}\backslash\{i,j\}$ 
  and compute 
  \begin{eqnarray*}
  && \sum_{\bar{S}\in\{1,\dots,r\}\backslash\{i,j\}\,:\,\#\bar{S}=s-2} \,\,\bar{x}_i\cdot\bar{x}_j\cdot \prod_{h\in \bar{S}} \bar{x}_h \\
  &=& \sum_{\bar{S}\in\{1,\dots,r\}\backslash\{i,j\}\,:\,\#\bar{S}=s-2} \,\,\left(x_ix_j+\left(\frac{x_i-x_j}{2}\right)^2\right)\cdot \prod_{h\in \bar{S}} x_h\\ 
  &\ge& \sum_{\bar{S}\in\{1,\dots,r\}\backslash\{i,j\}\,:\,\#\bar{S}=s-2} \,\,x_i\cdot x_j\cdot \prod_{h\in \bar{S}} x_h.
  \end{eqnarray*}
  Thus, we have $f(\bar{\mathbf{x}})\ge f(\mathbf{x})$. Next we remark that we have equality iff $\prod_{h\in \bar{S}} x_h=0$ for all 
  subsets $\bar{S}\in\{1,\dots,r\}\backslash\{i,j\}\,:\,\#\bar{S}=s-2$, i.e., there are most $s-3$ indices $h\in \{1,\dots,r\}\backslash\{i,j\}$ with 
  $x_h\neq 0$, so that $f(\mathbf{x})=0$, which clearly is not an optimal solution. Thus, in an optimal solution $\mathbf{x}$ all entries have to be equal. 
  Since $\sum_{i=1}^r x_i=m$ we obtain $x_i=\tfrac{m}{r}$ and the stated target value is a direct conclusion.     
  
  For the case with integral variables we assume that $\mathbf{x}=(x_1,\dots,x_r)$ is an optimal solution such that there exist indices $1\le i,j\le r$ with 
  $x_i-x_j\ge 2$. Now let $\bar{\mathbf{x}}$ arose from $\mathbf{x}$ by increasing $x_j$ and decreasing $x_i$ by one, respectively. Since $\mathbf{x}\in\mathbb{N}^r$ and 
  $x_i-x_j\ge 2$, also $\bar{\mathbf{x}}\in\mathbb{N}^r$ and $\sum_{h=1}^r \bar{x}_h=\sum_{h=1}^r x_h=m$. Next we will show $f(\bar{\mathbf{x}})\ge f(\mathbf{x})$. To this end, 
  we proceed as before and distinguish the summands in $\sum_{S\subseteq \{1,\dots,r\}\,:\,\# S=s} \,\,\prod_{i\in S} x_i$ and 
  $\sum_{S\subseteq \{1,\dots,r\}\,:\,\# S=s} \,\,\prod_{i\in S} \bar{x}_i$ according to the cardinality of $S\cap\{i,j\}$. As before, for $\# S\cap\{i,j\}\le 1$ there is 
  no difference if we compare the sum over all respective subsets $S$. For the cases $\# S\cap\{i,j\}=2$ we can utilize the inequality     
  $$
    (x_i-1)\cdot (x_j+1)\cdot z=x_ix_jz+(x_i-x_j-1)\cdot z\ge x_ix_jz
  $$
  for $z\ge 0$  to conclude $f(\bar{\mathbf{x}})\ge f(\mathbf{x})$. Thus, there exists an optimal solution $\mathbf{x}$ with $\left|x_i-x_j\right|\le 1$ for all $1\le i,j\le r$. Due 
  to symmetry we can assume $x_1\le\dots\le x_r$ w.l.o.g. Since $\sum_{i=1}^r x_i=m$, we obtain the stated formula $x_i=\left\lfloor\tfrac{m+i-1}{r}\right\rfloor$ for $1\le i\le r$. 
\end{proof}

\begin{proposition}
  \label{prop_first_improved_upper_bound}
  Let $C$ be a linear $[k+t,k]_2$ code and $\mathbf{a}$ its corresponding vector counting the multiplicities of the occurring information vectors. With this, we have
  $$
    M(C)\le \frac{(k+1)k}{2}+\sum_{s=2}^{t+1} {{2^t-1}\choose s} \cdot \left(\frac{k}{2^t-1}\right)^s.
  $$
\end{proposition}
\begin{proof}
  We want to apply Theorem~\ref{thm_M_C_a_formula} and remark that we clearly have
  $$
    k+\sum_{\tau\in\mathbb{F}_2^t\backslash\{\mathbf{0}\}} {{a_\tau}\choose 2} \le \frac{(k+1)k}{2}.
  $$
  Since no minimal generating set of cardinality at least two contains the zero vector and the $a_{\tau}$ are non-negative, we conclude  
  \begin{equation}
    \sum_{\hat{S}\subseteq\mathbb{F}_2^t\,:\,\hat{S}\text{ is minimal generating and } 2\le \#\hat{S}\le t+1} 
    \,\,\prod_{\tau\in\hat{S}} a_\tau \le 
    \sum_{S\subseteq\mathbb{F}_2^t\backslash\{\mathbf{0}\}\,:\,2\le \# S\le t+1}\,\, 
    \prod_{\tau\in S} a_\tau.\label{ie_tail}
  \end{equation}    
  Since $\sum_{\tau\in\mathbb{F}_2^t} a_{\tau}=k$ we can assume $a_{\mathbf{0}}=0$ when maximizing the right-hand side of Inequality~(\ref{ie_tail}). 
  Applying Lemma~\ref{lemma_symmetric_function_optimization} onto the right-hand side of Inequality~(\ref{ie_tail}), with $s=\# S$, $r=2^t-1$, and $m=k$, 
  gives the stated upper bound for $M(C)$.
\end{proof}

We remark that Proposition~\ref{prop_first_improved_upper_bound} improves upon the matroid upper bound $M_2(k+t,k)\le {{k+t}\choose{t+1}}$. As an example we state that 
Proposition~\ref{prop_first_improved_upper_bound} yields 
\begin{eqnarray*}
  M_2(k+2,k) &\le& \frac{k^3}{27}+\mathcal{O}\!\left(k^2\right),\\
  M_2(k+3,k) &\le& \frac{5k^4}{343}+\mathcal{O}\!\left(k^3\right),\text{ and}\\
  M_2(k+4,k) &\le& \frac{1001k^5}{253125}+\mathcal{O}\!\left(k^4\right),\\
\end{eqnarray*} 
while ${{k+2}\choose{2+1}}=\frac{k^3}{6}+\mathcal{O}\!\left(k^2\right)$, ${{k+3}\choose{3+1}}=\frac{k^4}{24}+\mathcal{O}\!\left(k^3\right)$, and 
${{k+4}\choose{4+1}}=\frac{k^5}{120}+\mathcal{O}\!\left(k^4\right)$. Note however that the fraction between the coefficients of the leading terms tend to $1$ as 
$t$ tends to infinity. In order to obtain tighter bounds we need to study the properties of minimal generating sets. 

\begin{lemma}
  \label{lemma_minimal_generating_card_2}
  For two different elements $a,b\in \mathbb{F}_2^t\backslash\{\mathbf{0}\}$ the set $\{a,b\}$ is minimal generating iff $\supp(a)\cap \supp(b)\neq\emptyset$.
\end{lemma}
\begin{proof}
  Note that we have $a+b\neq \mathbf{0}$. Since $b\neq\mathbf{0}$ the statement follows from the equivalence $\supp(a)\subseteq \supp(a+b)$ iff $\supp(a)\cap \supp(b)\neq\emptyset$.
\end{proof}

As an application of Theorem~\ref{thm_M_C_a_formula} we compute $M(C)$ in dependence of $\mathbf{a}$ for $t=2$.
\begin{proposition}
  \label{prop_a_formula_t_2}
  Let $C$ be a linear $[k+2,k]_2$ code and $\mathbf{a}$ its corresponding vector counting the multiplicities of the occurring information vectors. With this, we have
  \begin{eqnarray*}
    M(C)&=&k+\frac{a_{10}\cdot(a_{10}-1)}{2}+\frac{a_{01}\cdot(a_{01}-1)}{2}+\frac{a_{11}\cdot(a_{11}-1)}{2}\\ 
    && +a_{10}\cdot a_{11}+a_{01}\cdot a_{11}+a_{10}\cdot a_{01}\cdot a_{11}\\ 
    &=& k+\frac{(k-a_{00})\cdot(k-a_{00}-1)}{2}- a_{10}\cdot a_{01}+a_{10}\cdot a_{01}\cdot a_{11}.
  \end{eqnarray*}
\end{proposition}
\begin{proof}
  Due to Lemma~\ref{lemma_minimal_generating_card_2} the set $\{10,01\}$ is the only subset of $\mathbb{F}_2^2\backslash\{\mathbf{0}\}$ that has cardinality $2$ and is not 
  minimal generating. The unique subset $\{01,10,11\}$ of $\mathbb{F}_2^2\backslash\{\mathbf{0}\}$ of cardinality $3$ is indeed minimal generating. For the second equation 
  note that $k=a_{00}+a_{01}+a_{10}+a_{11}$. 
\end{proof}

Maximizing the formula from Proposition~\ref{prop_a_formula_t_2} we obtain: 
\begin{proposition}
  We have
  $$M_2(k+2,k)=k+k(k-1)/2+\lfloor(k-1)/3\rfloor \cdot \lfloor k/3\rfloor \cdot \lfloor(k+1)/3\rfloor$$
  for all $k\ge 1$.
\end{proposition}
\begin{proof}
  Let $C$ be a $[k+2,k]_2$ code.  From the latter expression for $M(C)$ in Proposition~\ref{prop_a_formula_t_2} it is obvious that $a_{00}=0$ and $a_{11}\ge 1$ in the maximum. 
  Thus, it remains to maximize 
  $$
    f\!\left(a_{01},a_{10},a_{11}\right)=a_{10}\cdot a_{01}\cdot a_{11}- a_{10}\cdot a_{01}=a_{10}\cdot a_{01}\cdot\left(a_{11}-1\right)
  $$
  subject to $a_{01}+a_{10}+a_{11}=k$ and $a_{01},a_{10},a_{11}\in\mathbb{N}$. It is well known that $f$ is maximized iff $a_{01}$, $a_{10}$, and $a_{11}-1$ are as 
  equal as possible while satisfying $a_{01}+a_{10}+\left(a_{11}-1\right)$, c.f.\ Lemma~\ref{lemma_symmetric_function_optimization}. Thus, an optimal solution is given by 
  $a_{01}=\left\lfloor\frac{k-1\,+\,2}{3}\right\rfloor$, $a_{10}=\left\lfloor\frac{k-1\,+\,1}{3}\right\rfloor$, and 
  $a_{11}=\left\lfloor\frac{k-1}{3}\right\rfloor+1$. Plugging into the formula in Proposition~\ref{prop_a_formula_t_2} gives the stated result.
\end{proof}

\begin{proposition}
  \label{prop_a_formula_t_3}
  Let $C$ be a linear $[k+3,k]_2$ code and $\mathbf{a}$ its corresponding vector counting the multiplicities of the occurring information vectors. With this, we have
  \begin{eqnarray*}
    M(C)&=&k+\sum_{\tau\in\mathbb{F}_2^3\backslash\{\mathbf{0}\}} \frac{a_{\tau}\cdot\left(a_{\tau}-1\right)}{2}
    +a_{110}\cdot\left(a_{101}+a_{011}+a_{111}\right)+a_{101}\cdot\left(a_{011}+a_{111}\right)+a_{011}\cdot a_{111}\\ 
    && +a_{100}\cdot \left(a_{110}+a_{101}+a_{111}\right)+a_{010}\cdot \left(a_{110}+a_{011}+a_{111}\right)+a_{001}\cdot \left(a_{011}+a_{101}+a_{111}\right) \\ 
    && +a_{100}a_{010}a_{110}+a_{100}a_{001}a_{101}+a_{010}a_{001}a_{011}  +a_{100}a_{010}a_{111}+a_{100}a_{001}a_{111}\\ &&+a_{010}a_{001}a_{111}
       +a_{100}a_{110}a_{011}+a_{100}a_{011}a_{101}
       +a_{010}a_{110}a_{101}+a_{010}a_{101}a_{011}\\
       &&+a_{001}a_{110}a_{101}+a_{001}a_{110}a_{011}
       +a_{100}a_{011}a_{111}+a_{010}a_{101}a_{111}+a_{001}a_{110}a_{111}\\&&+a_{110}a_{101}a_{011}
       +a_{110}a_{101}a_{111}+a_{110}a_{011}a_{111}+a_{011}a_{101}a_{111}\\ 
    && +a_{100}a_{010}a_{001}a_{111}+a_{100}a_{011}a_{110}a_{001}+a_{100}a_{101}a_{011}a_{010}+a_{100}a_{101}a_{110}a_{111}\\ 
    && + a_{010}a_{110}a_{101}a_{001}+a_{010}a_{110}a_{011}a_{111}+a_{001}a_{011}a_{101}a_{111} 
  \end{eqnarray*}
\end{proposition}
\begin{proof}
  We apply Theorem~\ref{thm_M_C_a_formula}. From the ${7\choose 2}=21$ $2$-subsets of $\mathbb{F}_2^3\backslash\{\mathbf{0}\}$ only the six subsets 
  $$
   \{100,010\}, \{100,001\}, \{010,001\}, \{100,011\}, \{010,101\}, \{001,110\}   
  $$ 
  violate the condition from Lemma~\ref{lemma_minimal_generating_card_2}. The 15 other combinations are listed in the first two rows of the stated formula. 
  It is a bit cumbersome to check by hand, but out of the ${7\choose 3}=35$ $3$-subsets of $\mathbb{F}_2^3\backslash\{\mathbf{0}\}$ just those $19$
  listed in the rows three to six of the stated formula satisfy the criterion of Lemma~\ref{lemma_reduced_code}. 
  The sum over the $7$ projective bases of $\mathbb{F}_2^3$ can be stated as
  $$
    \sum_{\tau_3\in\mathbb{F}_2^3\backslash\{\mathbf{0},\tau_1,\tau_2,\tau_1+\tau_2\}} a_{\tau_1}a_{\tau_2}a_{\tau_3} a_{\tau_1+\tau_2+\tau_3},
  $$
  see the subsequent Proposition~\ref{prop_leading_term_a}, and also be spelled out as done in the last two rows of the formula in the statement of the proposition.  
\end{proof} 

The exact maximization of the formula of Proposition~\ref{prop_a_formula_t_3} might be a technical challenge, while it is easy to come up with a conjecture 
for large enough values of $k$:
\begin{conjecture}
  \label{conj_m_2_t_3}
  For  $k\ge 4$ the exact value of $M_2(k+3,k)$ is given by the formula of Proposition~\ref{prop_a_formula_t_3} with 
  $\mathbf{a}=\left(a_{000},a_{100},a_{010},a_{001},a_{110},a_{101},a_{011},a_{111}\right)$, where 
  $$
    \mathbf{a}=\left\{\begin{array}{rcl}
    (l,l,l,l+1,l+1,l+1,l+1) &:& k=4+7l,\\
    (l,l,l,l+1,l+1,l+1,l+2) &:& k=5+7l,\\
    (l,l,l,l+1,l+1,l+2,l+2) &:& k=6+7l,\\
    (l,l,l,l+1,l+2,l+2,l+2) &:& k=7+7l,\\
    (l+1,l,l,l+2,l+2,l+1,l+2) &:& k=8+7l,\\
    (l+1,l,l,l+2,l+2,l+2,l+2) &:& k=9+7l,\\
    (l+1,l+1,l,l+2,l+2,l+2,l+2) &:& k=10+7l\\
    \end{array}\right.
  $$  
  if $k\le 26$ or 
  $$
  a_{000}=a_{001}=a_{110}=a_{111}=0, a_{100}=\left\lfloor\frac{k}{4}\right\rfloor, a_{010}=\left\lfloor\frac{k+1}{4}\right\rfloor, 
  a_{101}=\left\lfloor\frac{k+2}{4}\right\rfloor,\text{ and }a_{011}=\left\lfloor\frac{k+3}{4}\right\rfloor
  $$
  if $k\not\equiv 0\pmod 4$ and $k\ge 27$ or 
  $$
  a_{000}=a_{001}=a_{110}=a_{111}=0, a_{100}=\frac{k}{4}, a_{010}=\frac{k}{4}-1, 
  a_{101}=\frac{k}{4}+1,\text{ and }a_{011}=\frac{k}{4}
  $$
  if $k\equiv 0\pmod 4$ and $k\ge 27$.
\end{conjecture}  
We have computationally checked Conjecture~\ref{conj_m_2_t_3} for all $k\le 150$.  
For the leading term of $M_2(k+3,k)$, in terms of $k$, the situation is different to the one of 
Lemma~\ref{lemma_symmetric_function_optimization}, i.e., choosing $a_{000}=0$ and $a_\tau=\tfrac{k}{7}$ for $\tau\in\mathbb{F}_2^3\backslash\{\mathbf{0}\}$ just gives 
$M_2(k+3,k)\ge \frac{k^4}{343}+\mathcal{O}\!\left(k^3\right)$, while $a_{000}=a_{110}=a_{101}=a_{011}=0$ and $a_{100}=a_{010}=a_{001}=a_{111}=\tfrac{k}{4}$ gives 
$M_2(k+3,k)\ge \frac{k^4}{256}+\mathcal{O}\!\left(k^3\right)$ (ignoring the rounding to integers, whose effect is in $\mathcal{O}\!\left(k^3\right)$). 
Conjecture~\ref{conj_m_2_t_3} of course implies $M_2(k+3,k)= \frac{k^4}{256}+\mathcal{O}\!\left(k^3\right)$.

Next we focus on the leading term:
\begin{proposition}
  \label{prop_leading_term_a}
  Let $C$ be a linear $[k+t,k]_2$ code and $\mathbf{a}$ its corresponding vector counting the multiplicities of the occurring information vectors.
  If $t\ge 2$, then 
  $$
    M(C)=\mathcal{O}\!\left(k^t\right)+\frac{1}{(t+1)!}\cdot \sum_{\tau_1\in\mathcal{T}_1}\sum_{\tau_2\in\mathcal{T}_2}\dots \sum_{\tau_t\in\mathcal{T}_t} \left(\prod_{i=1}^t a_{\tau_i}\right) \cdot a_{\left(\sum_{i=1}^t \tau_i\right)}, 
  $$    
  where $\mathcal{T}_i=\mathbb{F}_2^t\backslash \left\langle\left\{\tau_j\,:\, 1\le j<i\right\}\right\rangle$ for $1\le i\le t$.
\end{proposition}
\begin{proof}
  We apply Theorem~\ref{thm_M_C_a_formula}. If $t\ge 2$ then only the contributions of the minimal generating sets $\hat{S}$ of cardinality exactly $t+1$ are not covered 
  by the $\mathcal{O}\!\left(k^t\right)$ term. Due to Corollary~\ref{cor_max_card_S} we have $\sum_{x\in\hat{S}}x=\mathbf{0}$ in those remaining cases. By Lemma~\ref{lemma_zero_sum_characterization} 
  we have to guarantee that no proper subset $\emptyset\neq\hat{T}\subsetneq \hat{S}$ satisfies $\sum_{x\in\hat{T}} x=\mathbf{0}$. Since there are $(t+1)!$ possible orders of the elements 
  of $\hat{S}$ we obtain the stated summation formula (which mimics the construction or counting of projective bases of $\mathbb{F}_2^t$).     
\end{proof}

We remark that the minimal generating sets of $\mathbb{F}_2^t$ of the maximum cardinality $t+1$ have a lot of equivalent descriptions. As mentioned before, they correspond 
to the projective bases of $\mathbb{F}_2^t$. Due to Corollary~\ref{cor_max_card_S} and Lemma~\ref{lemma_zero_sum_characterization} they also correspond 
to minimal dual codewords (of the $t$-dimensional simplex code). 

\begin{conjecture}
  \label{conj_leading_term}
  Let $t\ge 2$ be an integer and $\mathcal{P}=\left\{e_1,\dots,e_t,\mathbf{1}\right\}$. Then,
  the function $$\frac{1}{(t+1)!}\cdot \sum_{\tau_1\in\mathcal{T}_1}\sum_{\tau_2\in\mathcal{T}_2}\dots \sum_{\tau_t\in\mathcal{T}_t} \left(\prod_{i=1}^t a_{\tau_i}\right) \cdot a_{\left(\sum_{i=1}^t \tau_i\right)},$$ 
  where $\mathcal{T}_i=\mathbb{F}_2^t\backslash \left\langle\left\{\tau_j\,:\, 1\le j<i\right\}\right\rangle$ for $1\le i\le t$, attains its maximum on $\mathbb{R}_{\ge 0}^{2^t-1}$ subject to 
  the constraint $\sum_{\tau\in\mathbb{F}_2^t\backslash\{\mathbf{0}\}} a_\tau=k$ at $a_\tau=\tfrac{k}{t+1}$ for all $\tau\in\mathcal{P}$ and $a_\tau=0$ otherwise. If additionally 
  $a_\tau\in\mathbb{N}$ is assumed, then the maximum is attained at the points where $\left|a_{\tau}-a_{\tau'}\right|\le 1$ for all $\tau,\tau'\in\mathcal{P}$ and $a_\tau=0$ otherwise.     
\end{conjecture}
A direct implication of this conjecture is $M_2(k+t,k)=\left(\frac{k}{t+1}\right)^{t+1}+\mathcal{O}\!\left(k^t\right)$. For $t=2$ or $t=3$, $k\le 100$ Conjecture~\ref{conj_leading_term} 
is indeed true.
  
\section{Exact values for small parameters}

The aim of this subsection is to determine the exact value of $M_2(n,k)$ for cases with $1\le k\le n\le 15$. First note that if a linear code $C$ contains a codeword of weight $1$ 
then removing the corresponding coordinate yields a code $C'$ with $n(C') = n(C)-1$ and $M(C') = M(C)-1$. (In general we have $M(C)=M(C_1)+M(C_2)$ whenever $C=C_1\oplus C_2$, i.e., 
it is sufficient to consider indecomposable codes.) Removing zero or duplicate columns from the generator matrix of a binary code (scalar multiples for $q>2$) does not change 
the number of minimal codewords of the corresponding codes. Thus it is sufficient to consider all projective $[n,k]_2$ codes with minimum distance at least 2. 
These can be generated easily and for each code we can simply count the number of minimal codewords. To this end we have applied the enumeration algorithm from \cite{LC}, see Table~1 for 
the numerical results. In most cases we have verified the lower bounds from \cite{Maxmin2} to be exact and only improved the upper bounds. However, for $n=15$ there are also 
some improvements for the lower bounds. We remark that the rather complicated structure of the formula of $M_2(k+3,k)$ for $k\le 26$ in Conjecture~\ref{conj_m_2_t_3} suggests 
that the exact determination of $M_2(k+t,k)$ might not admit an easy explicit solution when $k$ is \textit{small}. 

\begin{table}[tbp]
\begin{center}
{\small
\begin{tabular}{|c|c|c|c|c|c|c|c|c|c|c|c|c|c|c|c|c|}
\hline
$n/k$ & 1 & 2 & 3 & 4 & 5 & 6 & 7 & 8 & 9 & 10 & 11 & 12 & 13 & 14 & 15\\ \hline
1 & 1 &   &   &  &  &  &  &  &  &  &  &  &  &  & \\ \hline
2 & 1 & 2 &   &  &  &  &  &  &  &  &  &  &  &  & \\ \hline
3 & 1 & 3 & 3 &  &  &  &  &  &  &  &  &  &  &  & \\ \hline
4 & 1 & 3 & 6 & 4 &  &  &  &  & &  &  &  &  &  & \\ \hline
5 & 1 & 3 & 6 & 10 & 5 &  &  &  & &  &  &  &  &  & \\ \hline
6 & 1 & 3 & 7 & 11 & 15 & 6 &  &  & &  &  &  &  &  & \\ \hline
7 & 1 & 3 & 7 & 14 & 17 & 21 & 7 &  & &  &  &  &  &  & \\ \hline
8 & 1 & 3 & 7 & 14 & 22 & 25 & 28 & 8 & &  &  &  &  &  & \\\hline
9 & 1 & 3 & 7 & 15 & 26 & 33 & 36 & 36 & 9 &  &  &  &  &  & \\ \hline
10 & 1 & 3 & 7 & 15 & 30 & 42 & 48 & 48 & 45 & 10 &  &  &  &  & \\ \hline
11 & 1 & 3 & 7 & 15 & 30 & 52 & 66 & 69 & 63 & 55 & 11 &  &  &  & \\ \hline
12 & 1 & 3 & 7 & 15 & 30 & 54 & 90 & 103 & 95 & 82 & 66 & 12 &  &  &\\ \hline
13 & 1 & 3 & 7 & 15 & 31 & 58 & 94 & 151 & 149 & 130 & 102 & 78 & 13 &  &\\ \hline
14 & 1 & 3 & 7 & 15 & 31 & 62 & 106 & 159 & 245 & 217 & 175 & 126 & 91 & 14 & \\ \hline
15 & 1 & 3 & 7 & 15 & 31 & 63 & 110 & 183 & 257 & 385 & 308 & 221 & 155 & 196 & 15 \\ \hline
\end{tabular}
}
\end{center}
\caption{$M_2(n,k)$ for $1\leq n\leq 15, 1\leq k\leq 15$}
\end{table}

\section*{Acknowledgments}
Romar dela Cruz would like to thank the Alexander von Humboldt Foundation for the support through the Georg Forster Research Fellowship, and the Mathematical Institute 
at the University of Bayreuth. Both authors benefit from discussions with Michael Kiermaier and Alfred Wassermann. Especially, the basic idea for the construction in 
Proposition~\ref{prop_projective_base_construction} is due to Michael Kiermaier.

\end{document}